\documentclass{llncs}

\usepackage{amsmath,amsfonts,amsbsy,amssymb,mathtools,tabulary,graphicx,times,caption,fancyhdr, mathrsfs, epsfig, epstopdf, lscape, subfig,booktabs, enumitem,lipsum,tkz-graph,scalerel,stackengine,setspace,adjustbox,float,geometry}

\begin{document}

\title{\large{Derangetropy in Probability Distributions and Information Dynamics}}

\author{Masoud Ataei$^1$ \and
Xiaogang Wang$^2$
}

\institute{Department of Mathematical and Computational Sciences, University of Toronto, Ontario, Canada\\ \email{masoud.ataei@utoronto.ca}
\and
Department of Mathematics and Statistics, York University, Ontario, Canada\\ \email{stevenw@yorku.ca}
}

\maketitle              

\begin{abstract}
We introduce derangetropy, a novel functional measure designed to characterize the dynamics of information within probability distributions. Unlike scalar measures such as Shannon entropy, derangetropy offers a functional representation that captures the dispersion of information across the entire support of a distribution. By incorporating self-referential and periodic properties, it provides deeper insights into information dynamics governed by differential equations and equilibrium states. Through combinatorial justifications and empirical analysis, we demonstrate the utility of derangetropy in depicting distribution behavior and evolution, providing a new tool for analyzing complex and hierarchical systems in information theory.
\end{abstract}

	\section{Introduction}
The accurate quantification and analysis of information are critical in many scientific disciplines, from data science and information theory to quantum mechanics and statistical physics. Traditionally, Shannon entropy has long served as the foundational measure of information, providing a scalar summary of uncertainty within a probability distribution \cite{Shannon1948}. While Shannon entropy's elegance and simplicity have ensured its widespread adoption, it also bears significant limitations, particularly when dealing with complex, non-Gaussian distributions or systems characterized by high-dimensional data.

A key limitation of Shannon entropy is its reduction of the functional space of distributions to a single scalar value, which can lead to identical entropy values for distinct distributions, thereby potentially obscuring their unique characteristics. This issue is particularly pronounced in analyzing non-Gaussian distributions, common in fields such as finance, genomics, and signal processing, where Shannon entropy often fails to capture intricate dependencies and tail behaviors critical to understanding information dynamics.

Moreover, Shannon entropy's reliance on a scalar summary frequently decouples it from other important statistical properties, such as variance, which may offer more direct insights into data spread and reliability \cite{petty2018some}. In certain cases, particularly involving nonlinear mappings, Shannon entropy can paradoxically suggest a reduction in information even when measurements enhance understanding of a variable, as it primarily emphasizes distribution flatness at the expense of other significant features \cite{cincotta2021shannon}.

To address these challenges, various extensions of entropy have been proposed, including Rényi entropy \cite{Renyi1961} and Tsallis entropy \cite{Tsallis1988}. While these measures have their merits, they remain fundamentally scalar-centric and do not fully address the need to capture the functional characteristics of information distribution across a distribution's entire support.

In response to these limitations, we propose \textit{derangetropy}, a novel conceptual framework that departs from traditional approaches to information measurement. Unlike existing measures, derangetropy is not merely an extension of entropy but represents a new framework that captures the dynamics of information across the entire support of probability distributions. By incorporating self-referential and periodic properties, derangetropy offers a richer understanding of information dynamics, particularly in systems where information evolve cyclically or through feedback mechanisms such as artificial neural networks.

This paper establishes the foundation of derangetropy by exploring its theoretical framework and properties. We begin in Section 2 by introducing the mathematical definition and key properties of derangetropy. Section 3 analyzes the information dynamics and equilibrium within this framework. In Section 4, we discuss the self-referential and self-similar nature of derangetropy. Section 5 provides a combinatorial perspective on derangetropy, connecting it to fundamental principles of combinatorial analysis. Finally, Section 6 summarizes the findings and suggests potential avenues for future research, including applications in information theory and beyond.

\section{Mathematical Definition and Properties}
\subsection{Derangetropy Functional}
Let $(\Omega, \mathscr{F}, \mathbb{P})$ be a probability space, where $\Omega$ denotes the sample space, $\mathscr{F}$ is a $\sigma$-algebra of measurable subsets of $\Omega$, and $\mathbb{P}$ is a probability measure on $\mathscr{F}$. Consider a real-valued random variable $X: \Omega \rightarrow \mathbb{R}$ that is measurable with respect to the Borel $\sigma$-algebra $\mathscr{B}_{\mathbb{R}}$ on $\mathbb{R}$. Assume that $X$ has an absolutely continuous distribution with a \textit{probability density function} (PDF) $f~\in~\mathcal{L}^2(\mathbb{R}, \mathscr{B}_{\mathbb{R}}, \lambda)$, where $\lambda$ denotes the Lebesgue measure on $\mathbb{R}$. The \textit{cumulative distribution function} (CDF) associated with random variable $X$ is given by
\begin{equation*}
	F(x) = \int_{-\infty}^x f(t) \, d\lambda(t), \quad x \in \mathbb{R}.
\end{equation*}

In the following, we present the formal definition of the derangetropy functional.
\begin{definition}[Derangetropy]
	\label{Def:Derangetropy_Sin}
	The \textit{derangetropy functional} $\rho: \mathcal{L}^2(\mathbb{R}, \mathscr{B}_{\mathbb{R}}, \lambda) \rightarrow \mathcal{L}^2(\mathbb{R}, \mathscr{B}_{\mathbb{R}}, \lambda)$ is defined by the following mapping
	\begin{equation}
		\rho[f](x) = \left(\frac{24}{\pi e}\right) \sin(\pi F(x)) F(x)^{F(x)} (1-F(x))^{1-F(x)}  f(x),
	\end{equation}
	where $f$ is the PDF associated with random variable $X$, and $\rho[f]$ refers to the function in $\mathcal{L}^2(\mathbb{R}, \mathscr{B}_{\mathbb{R}}, \lambda)$ obtained through this mapping. The evaluation of the derangetropy functional at a specific point $x \in \mathbb{R}$ is denoted by $\rho_f(x)$.
\end{definition}

The sine function's periodicity plays a crucial role in modeling systems characterized by cyclical or repetitive phenomena, where information alternates between concentrated and dispersed states. Specifically, the sine function serves as both a \textit{modulation mechanism} and a \textit{projection operator}. As a modulation mechanism, the sine function encodes the cyclical aspects of information content, capturing the periodic nature of information flow in systems with intrinsic cycles, such as those observed in time series data. By modulating the information content, the sine function effectively characterizes the periodic dynamics of information distribution within the probability space. Simultaneously, as a projection operator, the sine function translates the modulated information into a functional space where the dynamics of information flow become more apparent. This dual role facilitates the analysis of complex systems, enabling the representation of higher-dimensional informational structures and their influence on observable behavior.

To gain more insights into its structure, the derangetropy functional can be interpreted as a Fourier-type transformation as follows:
\begin{equation}
	\rho[f](x) = \left(\frac{24}{\pi \epsilon}\right) \sin(\pi F(x))\,  e^{-H_B(F(x))},
\end{equation}
where
\begin{equation}
	H_B(F(x)) = -F(x) \log F(x) - (1-F(x)) \log (1-F(x)),
\end{equation}
is the \textit{Shannon entropy} for a Bernoulli distribution with success probability $p = F(x)$. The term $H_B(F(x))$ quantifies the uncertainty or informational balance between the regions to the left and right of $x$, as indicated by the CDF values $F(x)$ and $1-F(x)$, respectively. In this interpretation, the derangetropy functional can be viewed as mapping the current informational energy space, represented by Shannon entropy, into a space characterized by oscillatory behavior. The exponential term $e^{-H_B(F(x))}$ modulates the sine function $\sin(\pi F(x))$, adding a layer of complexity that reflects how information is concentrated or dispersed across the distribution. This modulation is essential for understanding \textit{localized informational dynamics}, where the oscillations of the sine function mirror the underlying fluctuations in information content.

\subsection{Empirical Observations and Insights}
To illustrate the behavior of $\rho_f(x)$, we examine five representative distributions: uniform, normal, exponential, semicircle, and arcsin. Each distribution highlights different characteristics such as symmetry, skewness, tail behavior, and boundary effects, offering insights into the relationship between probability density functions and their corresponding derangetropy functionals.

As depicted in Figure \ref{fig:Derangetropy_Distributions}, symmetric distributions, like the normal and semicircle distributions, result in symmetric derangetropy functionals with prominent central peaks, indicating regions of high informational content. This symmetry suggests a strong correspondence between areas of high probability density and regions of significant information content, particularly in unimodal distributions where the central peaks of $\rho_f(x)$ aligns with the median, marking a key \textit{informational hotspot}.

In contrast, asymmetric distributions, such as the exponential distribution, produce skewed derangetropy patterns. These patterns reveal how the concentration of probability mass influences informational dynamics, with $\rho_f(x)$ shifting towards regions of higher probability density. This behavior is critical for understanding how information is distributed in systems characterized by non-uniform or skewed data.

Furthermore, the behavior of $\rho_f(x)$ near the boundaries of the distribution's support, particularly in distributions like the exponential and arcsin, emphasizes the importance of boundary dynamics in shaping the informational landscape. The pronounced derangetropy at the boundaries suggests that information content is significantly influenced by the extremities of the distribution. This observation is particularly relevant in fields like risk management or extreme value theory, where the accumulation and dissipation of information at the boundaries can provide critical insights into extreme events or rare occurrences.

Finally, the amplitude and frequency of oscillations in $\rho_f(x)$ are directly tied to the complexity of the underlying distribution. For distributions with uniform or smoothly varying density, such as the uniform or normal distributions, the oscillations in $\rho_f(x)$ are regular and predictable, indicating a steady flow of information. Conversely, distributions with more complex or non-uniform density, like the arcsin distribution, exhibit irregular oscillations, signaling that the flow of information is more dynamic and potentially chaotic. Thus, derangetropy can serve as a measure of the \textit{informational complexity} of a distribution.

\begin{figure}
	\centering
	\subfloat{{\includegraphics[scale=0.45]{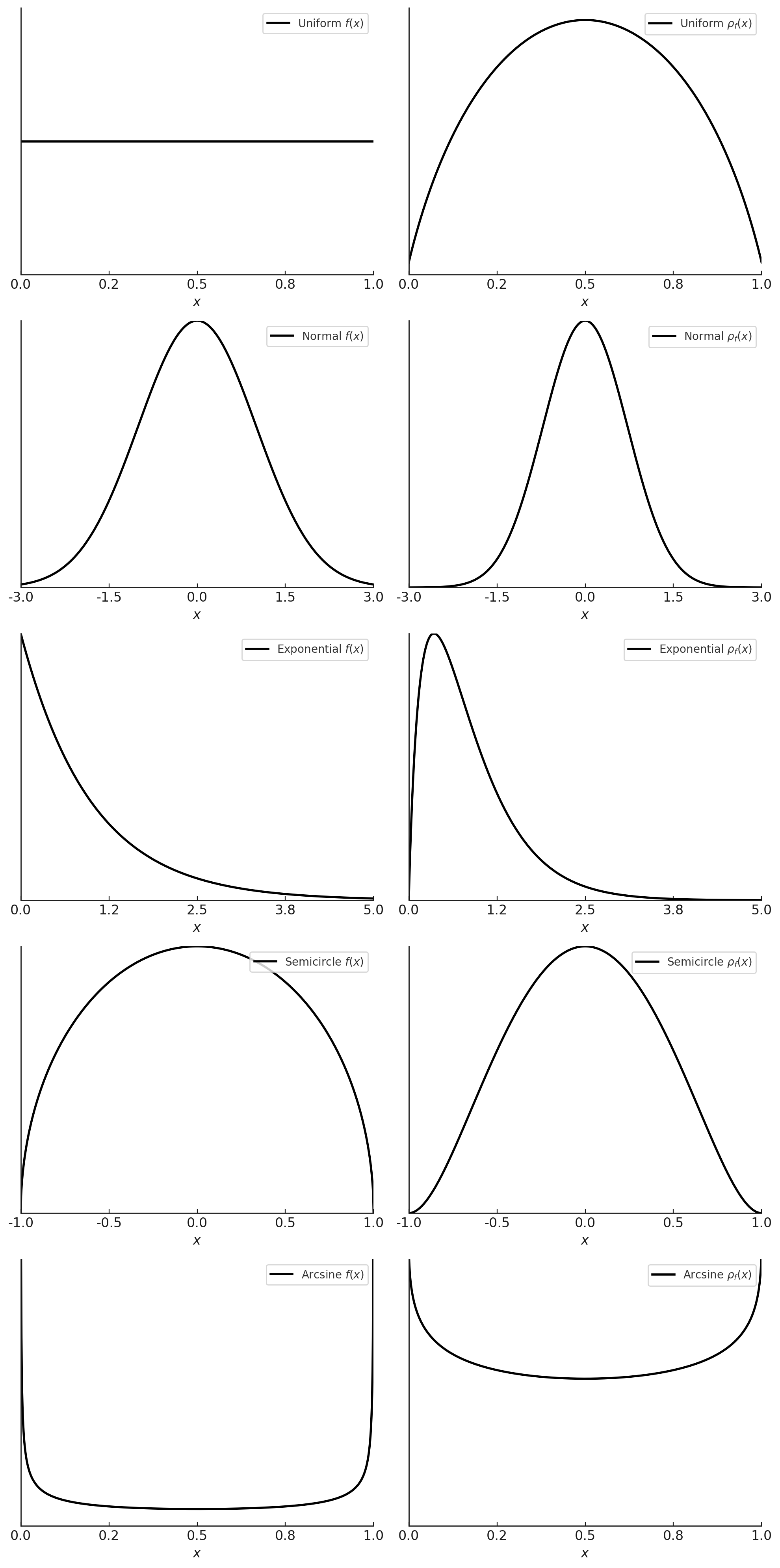} }}%
	\caption{Plots of probability density functions $f(x)$ (left) and derangetropy functionals $\rho_f(x)$ (right) for $\operatorname{Uniform}(0,1)$, $\operatorname{Normal}(0,1)$, $\operatorname{Exponential}(1)$, $\operatorname{Semicircle}(-1,1)$, and $\operatorname{Arcsin}(0,1)$ distributions.}%
	\label{fig:Derangetropy_Distributions}%
\end{figure}

\section{Information Dynamics and Equilibrium}

\subsection{Framework for Informational Energy}
The dynamics of informational content within a probability distribution can be likened to energy conservation principles in classical mechanics. However, the derangetropy functional reveals a more intricate interaction between different forms of informational energy, moving beyond strict conservation. Instead, the functional emphasizes a dynamic equilibrium between these energies, illustrating how information is distributed, concentrated, or dissipated across the distribution. This perspective provides deeper insights into the stability and evolution of information within complex systems.

The informational content described by the derangetropy functional $\rho_f(x)$ can be decomposed into distinct components that represent various aspects of information distribution. Mathematically, this decomposition is expressed as:
\begin{equation}
	\log \left(\rho_f(x)\right) = \log (\sin (\pi F(x))) - H_B(F(x)) + \log (f(x)) + C,
\end{equation}
where $C = \log \left(\frac{24}{\pi e}\right)$ is a constant. The total informational energy $\mathrm{E}_{\text{Total }}$ within a distribution is the sum of two primary components: oscillatory and structural informational energies. These components interact dynamically, maintaining a form of equilibrium across the distribution:
\begin{equation}
	\mathrm{E}_{\text{Total }} = -\log \left(\rho_f(x)\right) = \mathrm{E}_{\text{Oscillatory }} + \mathrm{E}_{\text{Structural }} + C.
\end{equation}
On the one hand, the \textit{oscillatory informational energy} defined by
\begin{equation}
	\mathrm{E}_{\text{Oscillatory }}(x) = -\log (\sin (\pi F(x))),
\end{equation}
captures the dynamic, cyclical flow of information within the distribution, similar to kinetic energy in physics, as energy shifts between different states or locations within the system.  		
On the other hand, the \textit{structural informational energy} defined by
\begin{equation}
	\mathrm{E}_{\text{Structural }}(x) = H_B(F(x)) - \log (f(x))
\end{equation}
reflects the inherent stability and uncertainty of the distribution, analogous to potential energy in physics.

Using the uniform distribution defined on $(0, 1)$ for illustrative purposes, Figure \ref{fig:Derangetropy_Energy} visually demonstrates how the total informational energy $\mathrm{E}_{\text{Total }}(x)$ varies across the distribution, driven by the interaction between oscillatory and structural energies. The solid line in the figure represents the total informational energy, which exhibits a parabolic shape with a minimum at $x = 0.5$ and steep rises near the boundaries. This pattern indicates that the energy distribution is influenced by both the oscillatory and structural components, with higher energy values near the boundaries suggesting a concentration of information content in these regions, where the distribution is less stable.

Furthermore, the dotted line, representing the oscillatory energy, dominates near the boundaries $x = 0$ and $x = 1$. The sharp increase in this component near the edges reflects the intense, cyclic nature of the information content in these areas, where the distribution is most unstable. As $x$ approaches 0 or 1, the sine function approaches zero, leading to spikes in oscillatory energy, which in turn causes significant instability. Also, the dashed line represents the structural energy, which peaks around the median $x = 0.5$ of the distribution. This peak reflects the inherent stability and maximum uncertainty at the center of the distribution. The structural energy remains relatively flat across the distribution, except near the boundaries where it diminishes, reflecting the distribution's overall stability in its central region compared to the more volatile boundary regions.

This figure shows that the total energy $\mathrm{E}_{\text{Total}}$ is not conserved in the traditional sense; it varies due to the compensatory interaction between $\mathrm{E}_{\text{Oscillatory}}$ and $\mathrm{E}_{\text{Structural}}$. As oscillatory energy increases near the boundaries, structural energy's relative contribution diminishes, and vice versa near the center. This dynamic equilibrium highlights how information is distributed and transformed within the distribution, providing a visual validation of the theoretical concepts discussed.

\begin{figure}
	\centering
	\subfloat{{\includegraphics[scale=0.45]{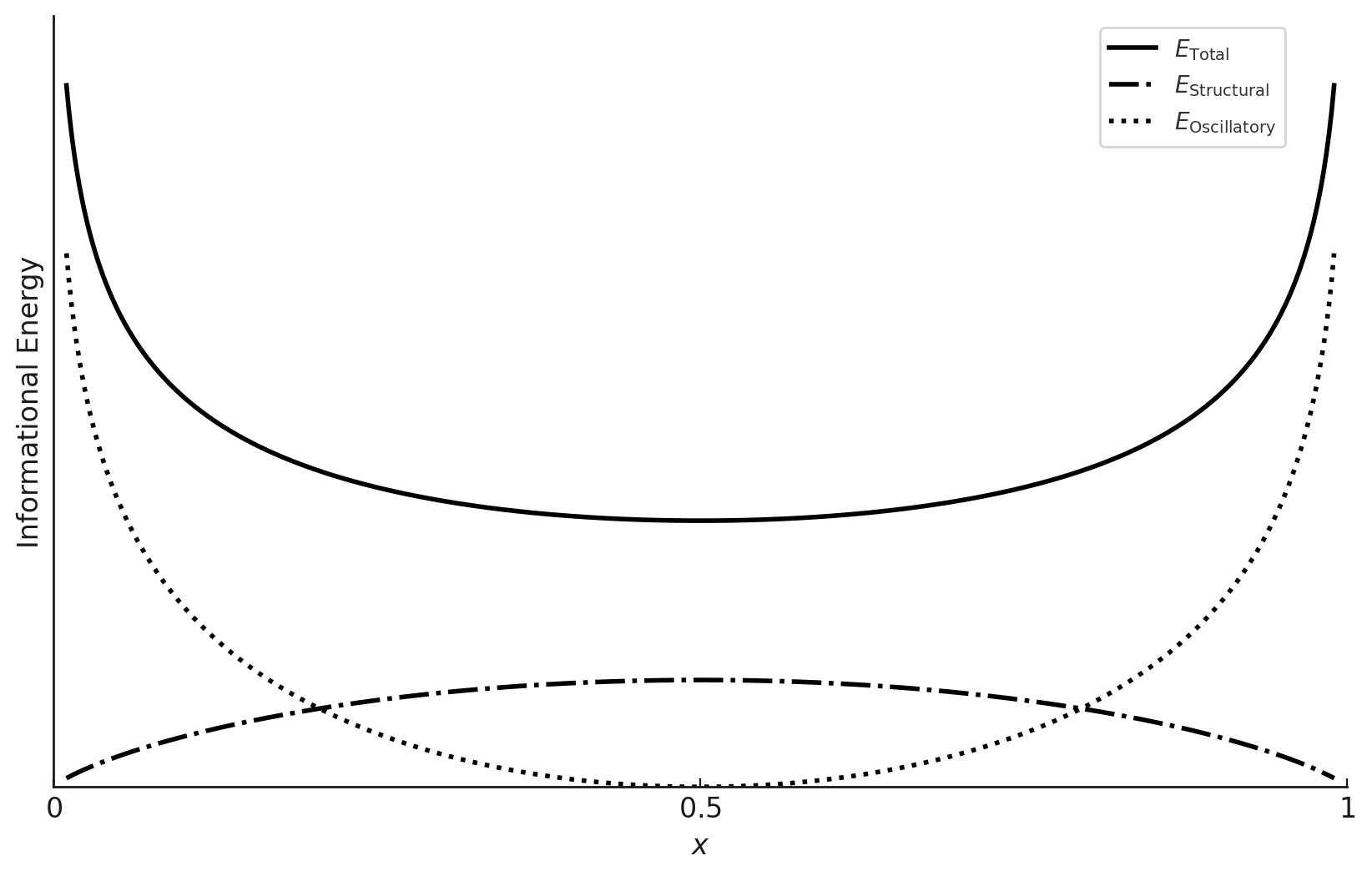} }}%
	\caption{Plots of total (solid line), structural (dashed line) and oscillatory (dotted line) informational energies for $\operatorname{Uniform}(0,1)$ distribution.}
	\label{fig:Derangetropy_Energy}%
\end{figure}
\subsection{Equilibrium Analysis}
Understanding equilibrium within the derangetropy framework requires analyzing the interaction between oscillatory and structural informational energies. Oscillatory energy, characterized by its fluctuations, tends to dominate near the boundaries of the distribution, where the probability density is low, and the informational content is more volatile. Conversely, structural energy is more pronounced near the center, reflecting the balance and uniformity of information in that region. This dynamic interplay forms the basis for understanding how equilibrium is achieved and maintained—or disrupted—within the distribution.

Equilibrium points within the derangetropy framework are critical for understanding the stability and dynamics of information distribution. These points are defined as locations where the derivative of the total informational energy with respect to $x$ is zero. The nature of these equilibrium points—whether they are stable, unstable, or saddle points—can be determined by examining the second derivative of the total informational energy.

For the uniform distribution, where $F(x) = x$ and $f(x) = 1$, the total informational energy is:
\begin{equation}
	\mathrm{E}_{\text{Total }}(x) = -\log (\sin (\pi x)) + H_B(x) + C.
\end{equation}
At $x = 0.5$, the equilibrium analysis reveals a critical point of instability. The minimum in the total energy curve at $x = 0.5$ suggests that this is a point of balance between the oscillatory and structural energies. However, this point is an unstable equilibrium, where small perturbations could cause the system to move away from equilibrium. Near the boundaries $x = 0$ and $x = 1$, the total energy increases steeply, indicating regions of instability. The oscillatory energy drives this instability, while the structural energy diminishes, making these points prone to significant shifts in the information distribution. This behavior is particularly relevant for understanding the distribution's tail properties, where extreme or rare events are more likely to occur.

\subsection{Key Mathematical Properties}
We now establish key mathematical properties of the derangetropy functional, $\rho_f(x)$, that underscore its utility in analyzing probability distributions. These properties provide insight into the behavior of $\rho_f(x)$ across different types of distributions and highlight its relevance to information dynamics. For any absolutely continuous PDF $f(x)$, the derangetropy functional $\rho[f](x)$ is a nonlinear operator that belongs to the space $C^{\infty}(\mathbb{R})$ having the following first derivative
\begin{equation}
	\label{Eq:Derivative}
	\frac{d}{dx} \rho_f(x) = \rho_f(x) \left[ \pi \cot(\pi F(x)) + \log\left(\frac{F(x)}{1-F(x)}\right) + \frac{f'(x)}{f(x)} \right],
\end{equation}
where the derivatives are taken with respect to $x$. The following theorem shows that the derangetropy functional involving distribution functions of random variable X, itself is a valid PDF for another random variable.
\begin{theorem}
	\label{Thrm:Normalization}
	For any absolutely continuous $f(x)$, the derangetropy functional $\rho_f(x)$ is a valid PDF.
\end{theorem}
\begin{proof}
	To prove that $\rho_f(x)$ is a valid PDF, we need to show that $\rho_f(x) \geq 0$ for all $x \in \mathbb{R}$ and that $\int_{-\infty}^{\infty} \rho_f(x) \, dx = 1$. The non-negativity of $\rho_f(x)$ is clear due to the non-negativity of the terms involved in its definition. Next, the normalization condition can be verified by the change of variables $z = F(x)$, yielding
	\begin{equation}
		\int_{-\infty}^{\infty} \rho_f(x) \, dx = \int_{0}^{1} \left(\frac{24}{\pi e}\right) \sin (\pi z) z^z(1-z)^{1-z} dz.
	\end{equation}	
	As shown in the Appendix, the integral	
	\begin{equation}
		\int_{0}^{1} \sin (\pi z) z^z(1-z)^{1-z} dz = \frac{\pi e}{24},
	\end{equation}	
	which, in turn, implies that	
	\begin{equation}
		\int_{-\infty}^{\infty} \rho_f(x) \, dx = 1.
	\end{equation}
	Hence, $\rho_f(x)$ is a valid PDF.
\end{proof}

Furthermore, since the derangetropy functional $\rho_f(x)$ is a valid PDF, it possesses well-defined theoretical properties, such as the existence of a mode. The following theorem demonstrates that, for any symmetric unimodal distribution, the mode of $\rho_f(x)$ coincides with the median of the underlying distribution.
\begin{theorem}
	For any symmetric unimodal probability distribution having PDF $f(x)$, the derangetropy functional $\rho_f(x)$ is maximized at the median of the distribution.
\end{theorem}
\begin{proof}			
	We aim to demonstrate that $\rho_f(x)$ attains its maximum at the median $m$, where $F(m)=0.5$. Given the symmetry of the distribution, the median $m$ coincides with the mode, and the PDF $f(x)$ is symmetric around $m$, implying that $f^{\prime}(m)=0$. By analyzing the derivative of the derangetropy functional, it follows that $x=m$ is a critical point of $\rho_f(x)$.
	To confirm that $m$ is indeed the point of maximum, we examine the second derivative of $\rho_f(x)$ at $x=m$:
	\begin{equation}
		\frac{d^2 \rho_f(x)}{d x^2}=\rho_f(m)\left[-\pi^2 \csc ^2(\pi F(m))+\left.\frac{d^2}{d x^2} \log \left(\frac{F(x)}{1-F(x)}\right)\right|_{x=m}+\frac{f^{\prime \prime}(m)}{f(m)}\right]
	\end{equation}				
	Since $f^{\prime \prime}(m) \leq 0$ due to the unimodal nature of the distribution, the second derivative is negative at $x=m$, confirming that $\rho_f(x)$ has a local maximum at this point. The symmetry of the distribution ensures that this local maximum at $m$ is, in fact, the global maximum of the derangetropy functional. Thus, $\rho_f(x)$ is maximized at the median $m$ of any symmetric unimodal distribution.efore, $\rho_f(x)$ is maximized at the median of any symmetric unimodal distribution.
\end{proof}

\subsection{Differential Equation}

The analysis of information dynamics within a probability distribution, particularly within the framework of the derangetropy functional, leads to the derivation of a nonlinear second-order ordinary differential equation that governs the evolution of informational content given by the following theorem.

\begin{theorem}
	Let $X$ be a random variable following a uniform distribution on the interval (0,1). Then, the derangetropy functional $\rho_f(x)$ satisfies the following second-order linear ordinary differential equation:
	\begin{equation}
		\frac{d^2 \rho_f(x)}{d F(x)^2}+4 \operatorname{atanh}(1-2 F(x)) \frac{d \rho_f(x)}{d F(x)}+\left[\pi^2-\frac{1}{F(x)(1-F(x))}+4 \operatorname{atanh}^2(1-2 F(x))\right]\rho_f(x)=0,
	\end{equation}	
	where $\operatorname{atanh}(z)=\frac{1}{2} \log \left(\frac{1+z}{1-z}\right)$ and the initial conditions are given by
	$$
	\left.\rho_f(x)\right|_{F(x)=0}=0, \quad \text { and }\left.\quad \frac{d \rho_f(x)}{d F(x)}\right|_{F(x)=0}=\frac{24}{e}.
	$$
\end{theorem}
\begin{proof}
	To eliminate the first-order derivative term, we multiply the entire equation by the integrating factor $\mu(F(x))$, where:	
	\begin{equation}
		\mu(F(x))=e^{\int-2 \operatorname{atanh}(1-2 F(x)) d F(x)}=(F(x)-1) e^{-2 \operatorname{atanh}(1-2 F(x))} .
	\end{equation}
	Next, introduce a new function $v(F(x))$ such that:
	\begin{equation}
		\rho_f(x)=\mu(F(x)) v(F(x))
	\end{equation}
	Substituting this expression into the differential equation simplifies it to:
	\begin{equation}
		\frac{d^2 v(F(x))}{d F(x)^2}+\pi^2 v(F(x))=0.
	\end{equation}
	The general solution to this simplified differential equation is:	
	\begin{equation}
		v(F(x))=C_1 \sin (\pi F(x))+C_2 \cos (\pi F(x))
	\end{equation}	
	where $C_1$ and $C_2$ are constants. Finally, evaluating the initial values yields $C_1=\frac{24}{\pi e}$ and $C_2=0$, implying that $\rho_f(x)$ is indeed a solution to the differential equation.
\end{proof}
This differential equation encapsulates the intricate dynamics of the derangetropy functional, driven by both linear and nonlinear influences.	At the median of the distribution, where $F(x)=0.5$, the function $\operatorname{atanh}(1-2 F(x))$ vanishes, leading to a stabilization in the rate of change of the informational content. This indicates that the median acts as a point of equilibrium, where the evolution of the derangetropy functional is locally stable. 

As $F(x)$ approaches the boundaries of the distribution, the inverse hyperbolic tangent function diverges, introducing significant nonlinearity into the equation. This divergence signals that the informational content becomes highly sensitive near the boundaries, where the probability mass is concentrated. The differential equation suggests that there is an accumulation or concentration of informational change at these boundaries. This is intuitive, as the cumulative distribution $F(x)$ compresses the probability mass into smaller regions near the boundaries, leading to more significant changes in the derangetropy functional. The boundary-sensitive term $\frac{1}{F(x)(1-F(x))}$, inversely proportional to the variance of a Bernoulli distribution with probability of success $F(x)$, further amplifies this effect, ensuring that the derangetropy functional adjusts dynamically in response to the distribution's proximity to its limits.

The stabilization at the median and the sensitivity near the boundaries highlight the complex interplay between linear and nonlinear dynamics within the differential equation. The coexistence of linear and nonlinear terms suggests that the evolution of the derangetropy functional is neither purely smooth nor entirely chaotic. Instead, it follows a complex pattern where smooth changes can be abruptly influenced by nonlinear effects, particularly near the boundaries and around the median. This interaction creates a dynamic equilibrium where the informational content is continually adjusted by the competing influences of linearity and nonlinearity. The derangetropy functional, therefore, adapts to different regions of the distribution, ensuring that the evolution of information is context-dependent.

The differential equation governing the derangetropy functional can also be interpreted through the framework of utility theory, offering a rigorous analogy between the evolution of informational content within a probability distribution and the concept of economic utility. As information begins to accumulate within the distribution, the utility—represented by the derangetropy functional—initially experiences rapid growth, analogous to the increasing satisfaction an economic agent derives from consuming additional units of a good. This phase of growth continues as the distribution approaches its median, where the utility reaches a peak. At this critical juncture, the informational content is optimally balanced, and the system achieves a state of maximum utility, reflecting the most efficient allocation of informational resources. The median serves as a natural equilibrium point, analogous to the optimal allocation of resources in economic theory, where utility is maximized through a balanced distribution of probability mass.

As the distribution progresses beyond the median and nears its boundaries, the utility derived from additional informational content begins to diminish, illustrating the principle of diminishing marginal utility in economics. The differential equation governing the derangetropy functional captures this transition, revealing that the rate of change in utility decreases as the cumulative distribution function approaches the extremes of the distribution. This decline in utility is akin to the reduced satisfaction an economic agent experiences when over-consuming a good, where the benefits of further consumption become increasingly marginal and may even turn negative. Consequently, the derangetropy functional not only models the accumulation and optimization of information but also encapsulates the risks associated with over-concentration. This interpretation underscores the functional's role in maintaining a balance analogous to the equilibrium sought in economic utility optimization, thereby providing a profound connection between information theory and economic principles.

\section{Self-referential Nature}
The derangetropy functional, denoted by \(\rho_f(x)\), is distinguished by its self-referential nature, whereby \(\rho_f(x)\) itself is a valid PDF. This implies that \(\rho_f(x)\) not only encapsulates the informational content of the underlying distribution \(f(x)\) but also recursively uncovers the hierarchical structure of its own information content. This recursive and hierarchical relationship is formalized as:
\begin{equation}
	\rho_f^{(n)}(x) = \left(\frac{24}{\pi e}\right) \sin\left(\pi \mathcal{G}_f^{(n-1)}(x)\right)\left(\mathcal{G}_f^{(n-1)}(x)\right)^{\mathcal{G}_f^{(n-1)}(x)}\left(1 - \mathcal{G}_f^{(n-1)}(x)\right)^{1 - \mathcal{G}_f^{(n-1)}(x)} \rho_f^{(n-1)}(x),
\end{equation}
where \(\rho_f^{(n)}(x)\) represents the \(n\)th iteration of the derangetropy functional, and $$\mathcal{G}_f^{(n)}(x) = \int_0^x \rho_f^{(n)}(t)\, dt$$ is the associated CDF. The initial conditions are set by 
$$\rho_f^{(0)}(x) = f(x) \quad \text{and} \quad \mathcal{G}_f^{(0)}(x) = F(x),$$ 
where \(F(x)\) denotes the CDF of the original distribution. This recursive structure ensures that each subsequent layer \(\rho_f^{(n)}(x)\) integrates the informational content of all preceding layers, thereby constructing a comprehensive hierarchical representation of information.

To elucidate this self-referential property, we consider the first and second layers of the derangetropy functional for the uniform distribution over the interval \((0, 1)\), as illustrated in Figure \ref{fig:Derangetropy_Layers_Uniform}. The first layer, \(\rho^{(1)}(x)\), reflects the basic informational structure of the uniform distribution. In contrast, the second layer, \(\rho^{(2)}(x)\), introduces additional complexity, manifesting in pronounced peaks and troughs. This evolution underscores the recursive amplification inherent in the derangetropy functional, progressively unveiling deeper informational layers with each iteration.

\begin{figure}[!htb]
	\centering
	\includegraphics[scale=0.45]{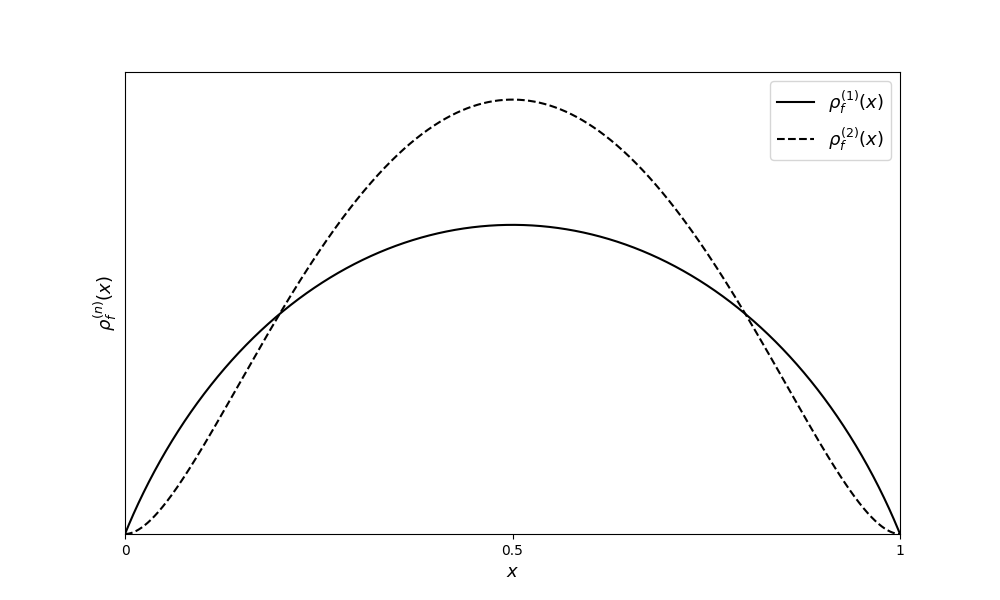}
	\caption{Plots of \(\rho_f^{(1)}\) (solid line) and \(\rho_f^{(2)}\) (dashed line) for $\operatorname{Uniform}(0,1)$ distribution.}
	\label{fig:Derangetropy_Layers_Uniform}
\end{figure}

\begin{figure}[!htb]
	\centering
	\includegraphics[scale=0.45]{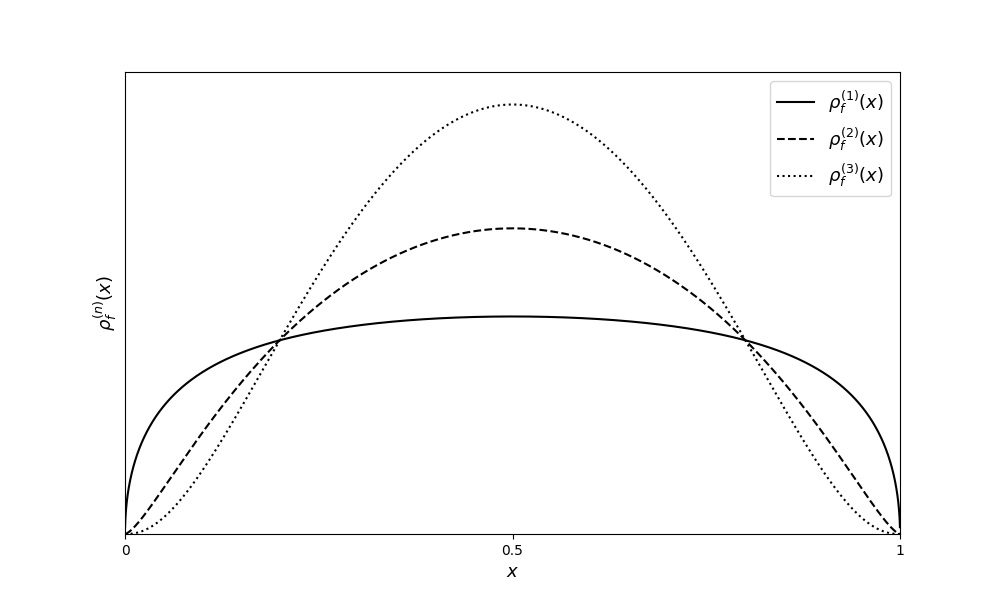}
	\caption{Plots of \(\rho_f^{(1)}\) (solid line), \(\rho_f^{(2)}\) (dashed line) and \(\rho_f^{(3)}\) (dotted line) for $\operatorname{Arcsin}(-1,1)$ distribution.}
	\label{fig:Derangetropy_Layers_Arcsin}
\end{figure}

Further insights into the self-referential nature of the derangetropy functional can be obtained by examining its behavior for the arcsin distribution, as shown in Figure \ref{fig:Derangetropy_Layers_Arcsin}. The first layer captures the strong boundary effects characteristic of the arcsin distribution, resulting in a pronounced concave upward profile near the interval’s endpoints. As recursion advances to the second layer, these boundary effects are amplified, giving rise to more pronounced oscillations and a complex informational structure. The transition in concavity observed between the first and second layers signifies a shift from a simple, boundary-focused representation to a more intricate depiction of the distribution’s information content.

The third layer introduces a bell-shaped curve, indicative of a stabilization process within the recursive structure. The previously amplified boundary effects become harmonized, redistributing the information more uniformly across the distribution. This bell-shaped curve suggests a tendency towards centralization and equilibrium within the distribution as recursion progresses. The smoothing of oscillations in the third layer reflects the derangetropy functional’s inherent ability to guide the system towards an informational equilibrium, where initially sharp boundary effects are moderated and the information distribution becomes more balanced and Gaussian-like.

An intriguing and nontrivial observation emerges when comparing the second- and third-level derangetropies of the arcsin distribution with the first- and second-level derangetropies of the uniform distribution. Notably, the second-level derangetropy of the arcsin distribution closely resembles the first-level derangetropy of the uniform distribution, while the third-level derangetropy of the arcsin distribution mirrors the second-level derangetropy of the uniform distribution. This resemblance reveals a deep structural connection between these distributions under the framework of the derangetropy functional.

The transformation \(X = \sin^2\left(\frac{\pi}{2} U\right)\) such that $U\sim \operatorname{Uniform}(0,1)$, which maps the uniform distribution to the arcsin distribution, is pivotal in explaining this connection. The sine function within the derangetropy functional plays a crucial role in this transformation. When applied to the arcsin distribution, whose CDF is given by \(\mathcal{G}_{\text{Arcsin}}(x) = \frac{2}{\pi} \arcsin(\sqrt{x})\), the sine function simplifies as follows:
\begin{equation}
	\sin\left(\pi \mathcal{G}_{\text{Arcsin}}(x)\right) = \sin\left(2 \arcsin(\sqrt{x})\right) = 2 \sqrt{x(1 - x)}.
\end{equation}
This expression directly mirrors the transformation \(X = \sin^2\left(\frac{\pi}{2} U\right)\), establishing a structural similarity between the arcsin and uniform distributions. The recursive nature of the derangetropy functional, which inherently tends to smooth and centralize information, reflects this underlying transformation. The arcsin distribution, characterized by its strong boundary effects, requires additional recursive layers to attain a similar informational structure to that of the uniform distribution. This explains why the second-level derangetropy of the arcsin distribution mirrors the first-level derangetropy of the uniform distribution, and why the third-level derangetropy of the arcsin distribution resembles the second-level derangetropy of the uniform distribution. Through its recursive application, the derangetropy functional reveals the deep structural relationships between these distributions.

Lastly, we note that as $n$ approaches infinity, the recursive process drives the distribution towards a degenerate distribution centered at the median. Mathematically, this means that the $nth$ level derangetropy $\rho_f^{(n)}(x)$ converges in distribution to a Dirac delta function $\delta(x-m)$, where all the probability mass is concentrated at the median $x=m$; i.e.,	
$$
\lim _{n \rightarrow \infty} \rho_f^{(n)}(x)=\delta(x-m)
$$
This convergence behavior is a direct consequence of the derangetropy functional's design, which inherently favors the centralization of information. The recursive smoothing effect ensures that, with each iteration, the distribution becomes more focused around the central point, ultimately leading to a situation where all mass is concentrated at the median. It is worth mentioning that the derangetropy functional exhibits self-similar patterns, particularly in the earlier stages of recursion. However, as the recursion deepens, the structure becomes more centralized about median and less fractal-like, leading to increasingly negative scaling exponents.

\section{Combinatorial Perspective}
The concept of derangetropy is deeply rooted in combinatorial principles, particularly the notion of derangements. In combinatorics, a derangement is a permutation of a set where no element appears in its original position, representing a state of maximal disorder. Extending this concept to continuous probability distributions, the derangetropy functional $\rho_f(x)$ quantifies the local interplay between order and disorder at each point in the distribution.

A key mathematical tool in deriving the properties of derangetropy is Euler's reflection formula, given by:
\begin{equation}
	\sin(\pi z) = \frac{\pi}{\Gamma(z)\Gamma(1 - z)}, \quad \text{for } z \notin \mathbb{Z},
\end{equation}
where $\Gamma(z)$ is the Gamma function, a continuous analogue of the factorial function. This formula establishes a fundamental connection between the sine function and the combinatorial structures underlying the derangetropy functional. The derangetropy functional, initially defined by formula \eqref{Def:Derangetropy_Sin}
can be reformulated using Euler's reflection formula as:
\begin{equation}
	\rho[f](x) = \left(\frac{24}{e}\right) \cdot F(x)^{F(x)} (1-F(x))^{1-F(x)} \cdot \frac{f(x)}{\Gamma(F(x))\Gamma(1 - F(x))}.
\end{equation}

This reformulation highlights the recursive self-referential nature of the distribution, where the distribution refers to its own values at each point $x$, thereby influencing the overall informational dynamics of the system. The terms $F(x)^{F(x)}$ and $(1-F(x))^{1-F(x)}$ introduce a self-weighting modulation mechanism that amplifies disorder non-linearly in regions of high probability, while regions with large $1-F(x)$ values experience a similar recursive influence. In continuous distributions, the combinatorial complexity is encapsulated by the Gamma functions $\Gamma(F(x))$ and $\Gamma(1-F(x))$. These functions serve as continuous analogues of the derangement process, quantifying the degree of disorder at each point by representing the combinatorial complexity associated with arranging elements less than $x$ or greater than $x$, respectively. The interplay between these Gamma functions at each point and the self-weighing mechanism determine the local derangetropy and reflects the underlying combinatorial structure of the distribution.

Moreover, the system governed by the derangetropy functional continuously transitions between states of order and disorder. This dynamic conversion is modulated by the oscillatory and structural informational energies inherent within the distribution, reflecting the probabilistic laws that govern the system's behavior. Equilibrium is achieved when the total informational energy $E_{\text{Total}}(x)$ is minimized across the distribution, representing a balance between order and disorder. Deviations from equilibrium increase derangetropy, indicating higher disorder at specific points $x$ within the distribution, a behavior crucial for understanding extreme events or rare occurrences.

Consequently, the derangetropy functional also captures both local variability and global stability within the system. At any given point $x$, the local derangetropy is determined by whether $x$ aligns with its expected order. Aggregating these local measures provides a global assessment of the system’s informational content, consistent with the principles of complex systems, where global properties emerge from local interactions. This dual capability makes derangetropy a powerful tool for analyzing both micro and macro structures within a wide range of probabilistic systems.

\section{Conclusion and Future Work}

This paper introduced \textit{derangetropy}, a functional measure designed to capture the dynamic nature of information within probability distributions. By offering a functional representation rather than a scalar summary, derangetropy provides deeper insights into the interplay between order and disorder across a distribution’s entire support. The measure's self-referential and periodic properties facilitate a rigorous analysis of information stability and evolution, enhancing our understanding of complex systems.

The versatility of derangetropy was demonstrated through its mathematical formulation and empirical analysis across various distributions. This novel approach extends beyond the capabilities of traditional entropy measures, allowing for a more detailed examination of information dynamics.

Looking forward, several promising research directions emerge. In the realm of deep neural networks, derangetropy could regulate information flow, preventing overconcentration or dispersion and thereby enhancing network robustness and generalization. Incorporating equilibrium relations into the backpropagation process may also introduce more efficient training methods.

In quantum information theory, derangetropy could provide new tools for analyzing information dynamics in quantum systems, offering insights where traditional measures fall short. Additionally, applying derangetropy within statistical mechanics and thermodynamics could deepen our understanding of complex systems at both equilibrium and nonequilibrium states.

In summary, derangetropy represents a significant conceptual and practical advancement in information theory, opening new avenues for theoretical exploration and practical application. Future research will likely expand its utility, revealing further insights into the dynamics of information and its potential to enhance the design and function of deep neural networks and other complex systems.

\newpage
\section*{Appendix}
\textbf{Proof of Theorem \ref{Thrm:Normalization}.}

Let us compute $\int_{0}^{1} \sin (\pi z) z^z(1-z)^{1-z}dz$ by defining
\begin{equation}
	\begin{aligned}
		S
		\, := \, &  \, \int_{0}^{1} z^z \, (1-z)^{1-z} \, e^{i\pi z} \, dz   \\
		\, = \, &  \, \int_{0}^{1} (1-z) \, \dfrac{z^z}{(1-z)^z} \, e^{i\pi z} \, dz   \\
		\, = \, &  \,  \int_{0}^{1} (1-z) \, \dfrac{e^{z\log z}}{e^{z\log(1-z)}} \, e^{i\pi z} \, dz    \\
		\, = \, &  \,  \int_{0}^{1} (1-z) \, e^{z( i\pi + \log z - \log(1-z) )} \, dz \, .   \\
	\end{aligned}
\end{equation}
By substituting $t=\log z - \log(1-z)$, we get $$ z=\dfrac{e^t}{e^t + 1} \, , $$ which in turn yields
\begin{equation}
	\begin{aligned}
		S
		\, = \, &  \, \int_{-\infty+i\pi}^{+\infty+i\pi} \dfrac{1}{e^t + 1} \, e^{(i\pi+t)\dfrac{e^t}{e^t + 1}} \, \dfrac{e^t}{(e^t+1)^2} \, dt   \\
		\, = \, &  \, \int_{-\infty+i\pi}^{+\infty+i\pi} e^{\dfrac{te^t}{e^t-1}} \, \frac{e^t}{(e^t-1)^3} \, dt \, .  \\
	\end{aligned}
\end{equation}
Now consider the following function  
\begin{equation}
	g(u) = e^{\dfrac{ue^u}{e^u-1}} \, \dfrac{e^u}{(e^u-1)^3} \, .
\end{equation}
This function is meromorphic on the strip $$ D = \{ u\in \mathbb{C}: -\pi \leq \Im(u) \leq \pi \} $$ and its only point is located at $u=0$ having
$$ Res(g,0) = -\frac{e}{24} \, . $$
Furthermore, consider the rectangular clockwise path $\alpha$, as depicted in Figure \ref{fig:Contour}, which is composed of $\alpha = \alpha_1 \oplus \alpha_2 \oplus \alpha_3 \oplus \alpha_4$, yielding
\begin{equation}
	\oint_{\alpha}g(u) du
	\, = \,   \, -2\pi i\, Res(g,0) \, = \,   \,\, 2 i \, \frac{\pi e}{24} \, . 
\end{equation}
\begin{figure}[!htb]
	\centering
	\includegraphics[scale=0.4]{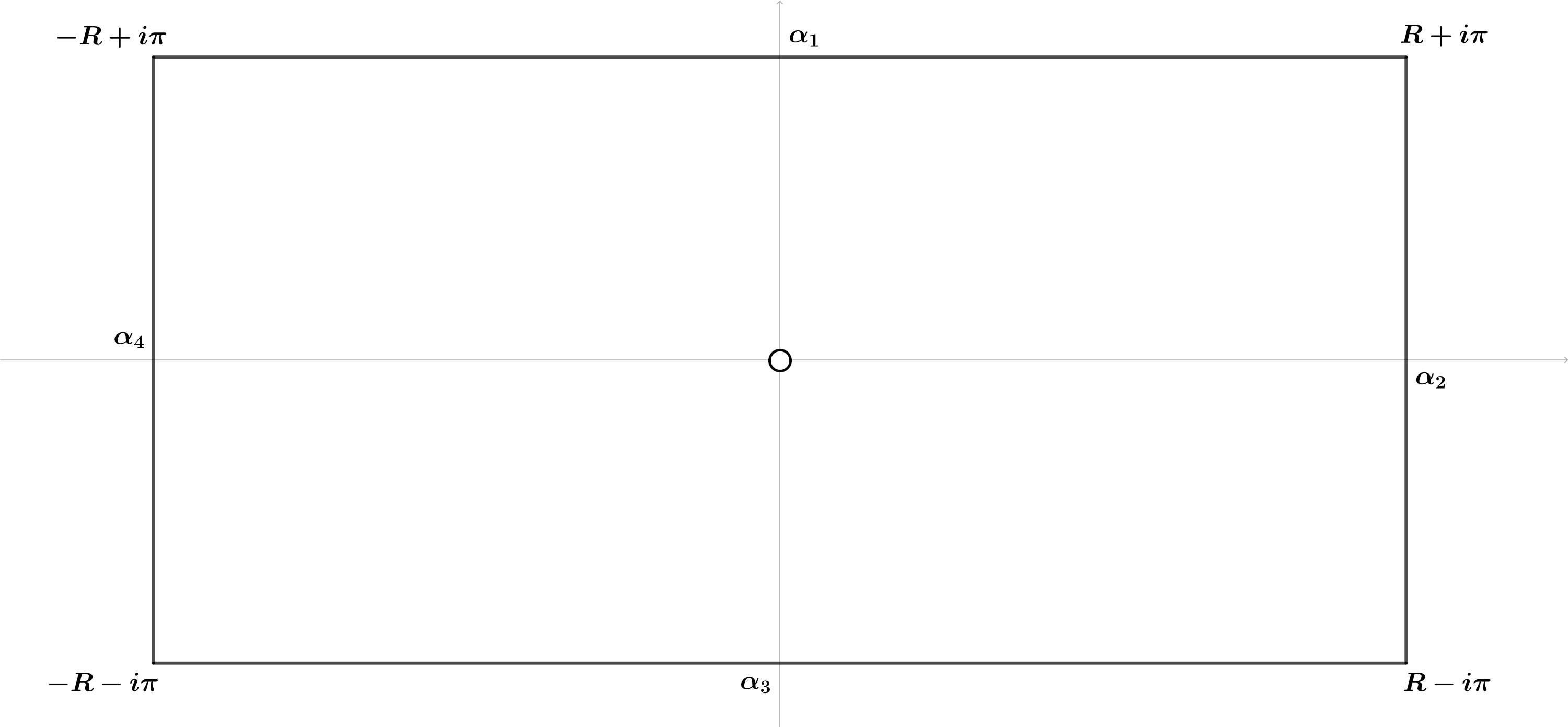}
	\caption{Schematic of the rectangular clockwise integration path.}
	\label{fig:Contour}
\end{figure}
Then, for any sequence $\{u_n\}\subset D$, we have that $|u_n|\to \infty$. In turn, this leads $g(u_n)$ to approach infinity. For this reason, $\int_{\alpha_2} g(u)du$ and $\int_{\alpha_4} g(u)du$ both become zero as $R\to \infty$. By further noting that $\int_{\alpha_1} g(u)du = \int_{\alpha_3} g(u)du$, one obtains $S=\int_{\alpha_1} g(u)du$ which yields $$\oint_{\alpha}g(u) du=2 i \frac{\pi e}{24}\cdot$$ Thus,
$$
\int_{0}^{1} \sin (\pi z) z^z(1-z)^{1-z}dz \, = \, \frac{\pi e}{24}.
$$

\bibliographystyle{splncs}      
\bibliography{cai}            

\end{document}